\documentclass[10pt,letter]{IEEEtran}
\usepackage[pass]{geometry}
\usepackage{nopageno}
\usepackage{etex}
\usepackage{subcaption}
\usepackage[font=footnotesize]{caption}
\usepackage{algorithm}
\usepackage{amsfonts}
\usepackage{amsthm}
\usepackage{amsmath}
\usepackage{amssymb}
\usepackage{steinmetz}
\usepackage{array}	
\usepackage{bbm}
\usepackage{cite}
\usepackage{dsfont}
\usepackage{epsfig}
\usepackage{float}
\usepackage[T1]{fontenc}
\usepackage{graphicx}

\usepackage{algorithmic}
\newcolumntype{L}[1]{>{\raggedright\let\newline\\\arraybackslash\hspace{0pt}}m{#1}}
\newcolumntype{C}[1]{>{\centering\let\newline\\\arraybackslash\hspace{0pt}}m{#1}}
\newcolumntype{R}[1]{>{\raggedleft\let\newline\\\arraybackslash\hspace{0pt}}m{#1}}

\newcommand{\algorithmfootnote}[2][\footnotesize]{%
  \let\old@algocf@finish\@algocf@finish
  \def\@algocf@finish{\old@algocf@finish
    \leavevmode\rlap{\begin{minipage}{\linewidth}
    #1#2
    \end{minipage}}%
  }%
}
\usepackage{url}
\usepackage{color}
\usepackage{blkarray} 
\usepackage{tikz}
\usetikzlibrary{calc}
\usepackage[export]{adjustbox} 
\usepackage{setspace}

\let\oldFootnote\footnote
\newcommand\nextToken\relax

\renewcommand\footnote[1]{%
    \oldFootnote{#1}\futurelet\nextToken\isFootnote}

\newcommand\isFootnote{%
    \ifx\footnote\nextToken\textsuperscript{,}\fi}

\usepackage[hang,flushmargin]{footmisc}
\usepackage{lipsum}
\usepackage{todonotes}

\newcommand{\bPhi}{\mbox{\boldmath{$\Phi$}}}

\newcommand{\bGamma}{\mbox{\boldmath{$\Gamma$}}}

{}
{}
\newtheorem{proposition}{Proposition}{}
{}
\newtheorem{corollary}{Corollary}{}
{}

\usepackage[normalem]{ulem} 

\usepackage{booktabs}

\newfont{\mycrnotice}{ptmr8t at 7pt}
\newfont{\myconfname}{ptmri8t at 7pt}

\IEEEoverridecommandlockouts
\allowdisplaybreaks

\begin{document}
\bstctlcite{IEEEexample:BSTcontrol}
\title{\Large  MU-Massive MIMO with Multiple RISs: SINR Maximization and Asymptotic Analysis}
%
%
%
%
\author{Somayeh~Aghashahi,~\IEEEmembership{Student Member,~IEEE,}
        Zolfa~Zeinalpour-Yazdi,~\IEEEmembership{Member,~IEEE,}
        Aliakbar~Tadaion,~\IEEEmembership{Senior Member,~IEEE,}
      Mahdi~Boloursaz Mashhadi,~\IEEEmembership{Member,~IEEE,}
        and Ahmed~Elzanaty,~\IEEEmembership{Senior Member,~IEEE}
\thanks{This work is based upon research funded by Iran National Science Foundation (INSF) under project No.4013201.}
     \thanks{S. Aghashahi, Z. Zeinalpour-Yazdi and A. Tadaion are with  the Department of Electrical Engineering, Yazd University, Yazd, Iran (e-mail: aghashahi@stu.yazd.ac.ir and \{zeinalpour, tadaion\}@yazd.ac.ir).}
     \thanks{M. B. Mashhadi  and A. Elzanaty  are with 5GIC \& 6GIC, Institute for Communication Systems (ICS), University of Surrey, Guildford, GU2 7XH, United Kingdom (e-mail: \{m.boloursazmashhadi, a.elzanaty\}@surrey.ac.uk). }
     }




        \maketitle

\begin{abstract}
In this letter, we investigate the signal-to-interference-plus-noise-ratio (SINR) maximization problem in a multi-user massive multiple-input-multiple-output (massive MIMO) system enabled with multiple reconfigurable intelligent surfaces (RISs). We examine two zero-forcing (ZF) beamforming approaches for interference management namely BS-UE-ZF and BS-RIS-ZF that enforce the interference to zero at the users (UEs) and the RISs, respectively.
Then, for each case,  we resolve the SINR maximization problem to find the optimal phase shifts of the elements of the RISs. 
Also, we evaluate the asymptotic expressions for the optimal phase shifts and the maximum SINRs when the number of the base station (BS) antennas tends to infinity. We show that if the channels of the RIS elements are independent and the number of the BS antennas tends to infinity,  random phase shifts achieve the maximum SINR using the BS-UE-ZF beamforming approach. The simulation results illustrate that by employing the BS-RIS-ZF beamforming approach, the asymptotic expressions of the phase shifts and maximum SINRs achieve the rate obtained by the optimal phase shifts even for a small number of the BS antennas.
\end{abstract}


%
\IEEEpeerreviewmaketitle

\section{Introduction}

\IEEEPARstart{R}{econfigurable} intelligent surface (RIS) is a new physical layer technology introduced to overcome the drawbacks of wireless communications such as fading, blockage, and interference \cite{wu2019towards}.  An RIS is a planner array of passive reflecting elements which can change the phase of the reflected signals of which the phase shifts can be continuously adjusted  \cite{survey_RIS}. Thus, based on the scenario, the behavior of the wireless environment can be improved by deploying the RISs in appropriate places and optimizing the phase shifts of their elements.  For instance, exploiting an RIS in a cell with one user can  solve the blockage problem or increase the SNR of the user equipment (UE) \cite{RIS_sigprocess}. However, in multi-user scenarios with a multiple input multiple output (MIMO) base station (BS), interference management by joint design of beamforming vectors of the BS and phase shifts of the RIS elements is a challenge. 

Several studies have been conducted on solving the challenges of beamforming and phase shift design in an RIS-aided multi-user MIMO scenario with different objectives \cite{rate1,rate4_weighted, SE1,EE1,EE2,power_min,max_min_SINR,subhash2022optimal,elzanaty,subhash2022max,AmanatElzanaty:21} such as maximization of the sum rate \cite{rate1,rate4_weighted}, spectral efficiency \cite{SE1}  and energy efficiency \cite{EE1,EE2} or minimization of the transmission power  \cite{power_min},  max-min SINR problem \cite{max_min_SINR,subhash2022optimal}, and the exposure to electromagnetic fields \cite{elzanaty,subhash2022max}. 
However, the proposed solutions involve intricate algorithms and straightforward designs are not presented in the literature. Moreover, the analysis of the results for a massive number of  BS antennas is not investigated.

Regarding the above considerations, in this letter, we investigate the SINR maximization problem in the downlink transmission of a multi-user massive MIMO system assisted with multiple RISs. In this scenario, we assume that some of the UEs are directly served by the BS, while others have no direct channel to the BS and are served through the RISs. Each RIS serves several UEs in a specific geographic area, forming a cluster. We use two zero-forcing (ZF) approaches for the beamforming of the BS.
The advantage of ZF beamforming is that it completely cancels the interference and this gets us a degree of freedom for asymptotic analysis and phase shift design.
Moreover, the ZF beamforming in massive MIMO systems achieves a near-optimal performance\cite{6736761}.
~The main contributions of this letter can be summarized as follows:
	(i) We employ a ZF beamforming approach named BS-UE-ZF  that nulls the interference at the UEs, and obtain the optimal  phase shifts of the RISs that maximize the SINR of the UEs. 
Then, for the case that one UE is connected to each RIS, we derive  asymptotic expressions for the optimal phase shifts  as the number of the BS antennas tends to infinity.
	(ii)  We prove that if the RISs experience i.i.d. normal channels,  random phase shifts achieve the maximum SINR  when the number of the BS antennas tends to infinity.  
(iii) For the case that one UE is connected to each RIS, we propose another ZF beamforming approach named BS-RIS-ZF that nulls the interference at the RISs. This  approach leads to closed-form expressions for the RIS phase shifts that maximize the SINRs at  UEs. 
(iv)  We evaluate asymptotic expressions for the optimal phase shifts and the maximum SINRs for the BS-RIS-ZF.
     \begin{figure}[t!]
     \centering
 \includegraphics[width=.75\columnwidth]{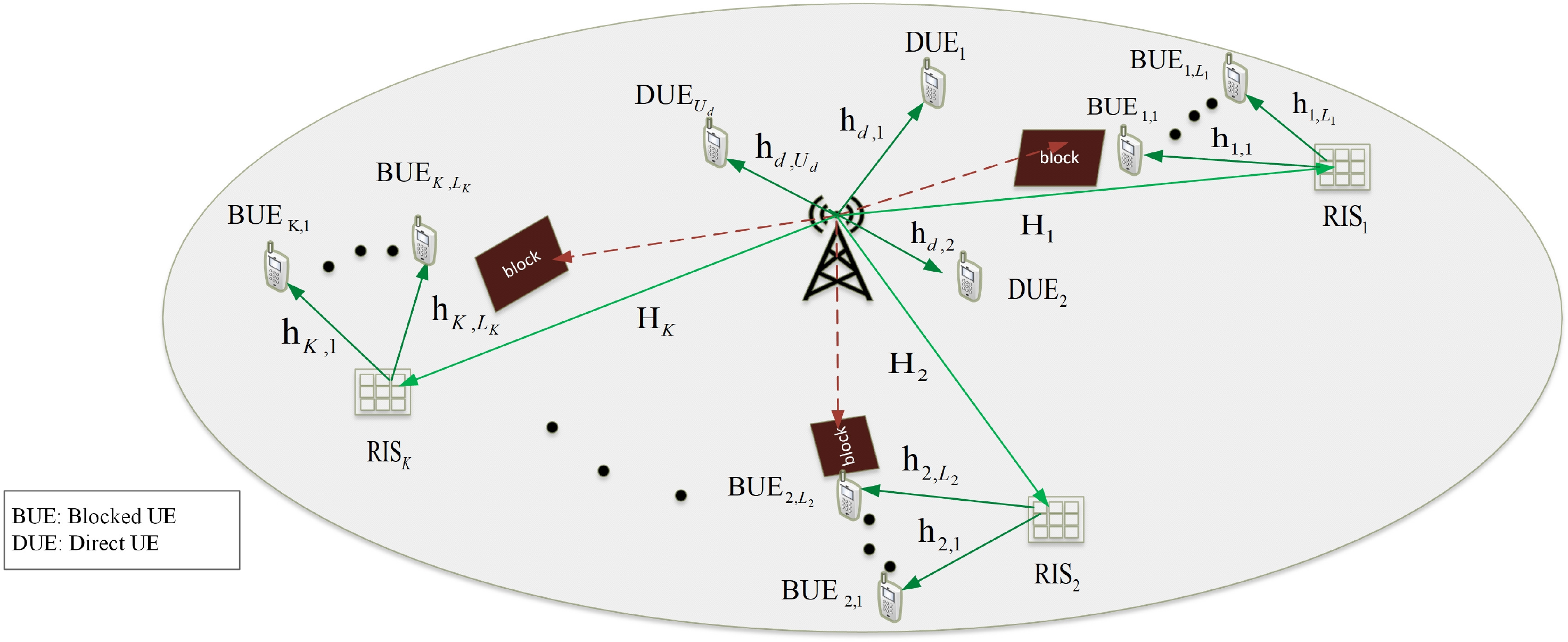}
 \caption{  The system model}
 \label{fig_sys_model}
    \end{figure} 
\section{System Model}
We consider  downlink transmission in a multi-RIS multi-user massive  MIMO  system. We   assume that  $U_d$  of the UEs (named direct UEs) are directly connected to the BS, and $U_b$ of them (named blocked UEs) have no direct connection with the BS  and  $K$ RISs are deployed to assist the communication between the BS and  these UEs. For $k\in \mathcal{K} \triangleq \{1,\dots,K\}$, $L_k$ of the blocked UEs are served through the $k^{\rm th}$ RIS which are selected by algorithm \ref{alg1}.
Moreover,  the channels between this RIS and the other UEs are blocked as considered in \cite{NOMA}.
Let,  $M$ be the number of the BS antennas and $N$ be the number of the elements at each of the RISs.  Also, $\mathbf{H}_{k} \in \mathbb{C}^{M\times N}$, $ \mathbf{h}_{k,\ell}\in \mathbb{C}^{N \times 1}$, and $ \mathbf{h}_{d,u}\in \mathbb{C}^{M \times 1}$   denote the channel matrix between the  BS and the $k^{\rm th}$ RIS, the channel vector between the $k^{\rm th}$ RIS and the $\ell^{\rm th}$ blocked UE connected to this RIS, and the channel between the BS and the $u^{\rm th}$  direct UE, respectively. We  assume that each of the channel coefficients has a complex normal distribution. The channel between the BS and each IRS is correlated such that $\mathbb{E}\{\mathbf{H}_{{k}}^{H}\mathbf{H}_{k}\}=\mathbf{R}_k$, while the channels between the BS and different RISs are independent, i.e., $\mathbb{E}\{\mathbf{H}_{{k}}^{H}\mathbf{H}_{k^{\prime}}\}=\mathbf{0}$ for $k\neq k^{\prime}$.
 The diagonal reflected matrix for the $k^{\rm th}$ RIS is  $\bPhi_{k} \triangleq \mathrm{diag}(e^{j\phi_{k,1}},\dots,e^{j\phi_{k,N}})$, where for $i \in  \mathcal{N}\triangleq \{1,\dots, N\}$, $\phi_{k,i}$ is the phase shift of the $i^{\rm th}$ element of the RIS. 
 Thus, the received signal at the $\ell^{\rm th}$ blocked UE connected to the $k ^{\rm th}$ RIS is given by
\begin{align}
& y_{_{b,k,\ell}}\hspace{-1mm} = \hspace{-1mm}\mathbf{h}_{k,\ell}^{H}\bPhi_{k}\mathbf{H}_{k}^{H}\hspace{-0.5mm}\mathbf{w}_{b,k,\ell} x_{_{b,k,\ell}}\hspace{-0.7mm} \notag
 +\hspace{-1.5mm}\sum_{\substack{ m=1\vspace*{-0.25mm} \\ m\neq k}}^{K}\hspace{-1mm}\sum_{\substack{ i=1}}^{L_m}\hspace{-0.5mm}   \mathbf{h}_{k,\ell}  ^{H}\bPhi_{k}\mathbf{H}_{k}^{H}\hspace{-0.5mm}\mathbf{w}_{b,m,i} x_{_{b,m,i}}
 \\ & \hspace{-1mm}+\hspace{-1mm}\sum_{\substack{ i=1\vspace*{-0.1mm}\\ i \neq \ell}}^{L_k}\hspace{-0.5mm}  \mathbf{h}_{k,\ell}  ^{H}\bPhi_{k}\mathbf{H}_{k}^{H}\mathbf{w}_{b,k,i} x_{_{b,k,i}}\hspace{-0.7mm}+\hspace{-0.7mm}\sum_{n=1}^{U_d}\mathbf{h}_{k,\ell}^{H}\bPhi_{k}\mathbf{H}_{k}^{H}\mathbf{w}_{d,n} x_{_{d,n}}
  \label{eq_block_sig} 
\hspace{-1mm}+\hspace{-0.7mm}z_{_{b,k,\ell}}, 
\end{align}
and the received signal at the $u^{\rm th}$ direct UE is 
\begin{align}
& y_{_{d,u}}\hspace{-1.4mm}=\hspace{-1mm}\mathbf{h}_{d,u}^{H}\hspace{-0.5mm}\mathbf{w}_{d,u} x_{_{d,u}} \hspace{-1mm}+\hspace{-1.8mm}\sum_{\substack{ n=1 \vspace*{-0.25mm}\\ n \neq u}}^{U_d}\hspace{-0.5mm}\mathbf{h}_{d,u}^{H}\hspace{-0.5mm}\mathbf{w}_{d,n} x_{_{d,n}} \hspace{-1mm} +\hspace{-2mm}\sum_{m=1}^{K}\hspace{-1mm}\sum_{\ell=1}^{L_m}\hspace{-0.5mm}\mathbf{h}_{d,u}^{H}\hspace{-0.5mm}\mathbf{w}_{b,m,\ell} x_{_{b,m,\ell}}\notag \\ &+ z_{_{d,u}}, \label{eq_direct_sig}
\end{align}
where $\mathbf{w}_{b,k,\ell}$ and $\mathbf{w}_{d,u}$ are the beamforming vectors, $x_{_{b,k,\ell}}$ and $x_{_{d,u}}$  are the transmit signals and $z_{_{b,k,\ell}}$ and $z_{_{d,u}}$ are the AWGN
  corresponding to the $\ell^{\rm th}$ blocked UE connected to the $k^{\rm th}$ RIS and the $u^{\rm th}$  direct UE, respectively. The first terms in \eqref{eq_block_sig} and \eqref{eq_direct_sig} are the desired signals for the $\ell^{\rm th}$ blocked and $u^{\rm th}$ direct UEs, respectively, while the second and third terms are the interference. Thus, the SINR of the $\ell^{\rm th}$ blocked UE connected to the $k^{\rm th}$ RIS  can be written as
\begin{align}
\label{SINR}
 & \mathrm{SINR}_{b,k,\ell}={\dfrac{|\mathbf{h}_{k,\ell}^{H}\bPhi_{k}\mathbf{H}_{k}^{H}\mathbf{w}_{b,k,\ell}|^2}{\mathcal{I}_{b,k,\ell}+\sigma^2_{k,\ell}  },} 
\\  \notag \hspace{-12mm} \textrm{where} \ & \mathcal{I}_{b,k,\ell}\hspace{-1mm}=\sum_{\substack{ m=1\vspace*{-0.25mm}\\ m\neq k}}^{K}\hspace{-1mm}\sum_{\substack{ i=1}}^{L_m}\hspace{-0.5mm}  | \mathbf{h}_{k,\ell}^{H}\bPhi_{k}\mathbf{H}_{k}^{H}\mathbf{w}_{b,m,i}|^2+\\   &  \hspace{-10mm}+\sum_{\substack{ i=1\vspace*{-0.1mm}\\ i \neq \ell}}^{L_k} | \mathbf{h}_{k,\ell}^{H}\bPhi_{k}\mathbf{H}_{k}^{H}\mathbf{w}_{b,k,i}\hspace{-0.2mm}|^2+ \hspace{-0.7mm}\sum_{n=1}^{U_d}|\mathbf{h}_{k,\ell}^{H}\bPhi_{k}\mathbf{H}_{k}^{H}\mathbf{w}_{d,n}|^2,\notag
\end{align}
\vspace{-2mm}
while the SINR of the $u^{\rm th}$direct UE is 
\begin{equation}
\label{SINR}
\mathrm{SINR}_{d,u}\hspace{-0.7mm}=\hspace{-0.7mm}\dfrac{|\mathbf{h}_{d,u}^{H}\mathbf{w}_{d,u}|^2}{\hspace{-0.5mm}\sum\limits_{\substack{n=1\vspace*{-0.25mm} \\ n \neq u}}^{U_d}\hspace{-0.5mm} |\mathbf{h}_{d,u}^{H}\mathbf{w}_{d,n}|^2\hspace{-0.5mm}+\hspace{-1mm}\sum\limits_{\substack{m=1}}^{K}\hspace{-0.5mm} \sum\limits_{\substack{\ell=1}}^{L_m} |\mathbf{h}_{d,u}^{H}\mathbf{w}_{b,m,\ell}|^2\hspace{-0.5mm}+\hspace{-0.5mm}\sigma^2_u  }.
\end{equation}

\begin{algorithm}[]
\caption{}
\label{alg1}
\algsetup{linenosize=\small}
\scriptsize
\begin{algorithmic}[1]
\STATE
 \textbf{Input}: \\ The allowed number of UEs to be scheduled, $U_{\max}$. \\
 The minimum value of required receive power, $p_{ \rm min}$.
\STATE 
\textbf{Output}: Scheduled UEs and assigned RISs.
\STATE Turn off the RISs and broadcast the signal.
\STATE Save the receive power of the UEs. \label{powBS}
\FOR{$k=1$ to $K$}
\STATE Turn on the $k^{\rm th}$ RIS and broadcast the signal.
\STATE Save the receive power of the UEs. \label{powIRS}
\STATE Turn off the $k^{\rm th}$ RIS.
\ENDFOR
\STATE Schedule up to $U_{\max}$ number of the UEs, whose maximum  receive power is more than $p_{\rm min}$ and assign the RISs \footnotemark.
\end{algorithmic}
\end{algorithm} 
\vspace{-0.4cm}
\section{SINR Maximization Problem}
In this section, first, we employ the BS-UE-ZF beamforming approach in which we cancel the interference at the UEs. Then, for the special case when one UE is related to each of the RISs, we design  the BS-RIS-ZF beamforming approach in which we cancel the interference at the RISs. For each of the beamforming methods, we design the phase shifts of the RISs to maximize the SINR  of the blocked UEs and  obtain  asymptotic  results when the number of  BS antennas tends to infinity, i.e., $M \longrightarrow \infty$.\footnotetext{The UEs  with the highest receive powers can be scheduled and each UE  can be assigned to the RIS that maximizes its received power.} 
\vspace{-0.35cm}
\subsection{ BS-UE-ZF Beamforming}
In the BS-UE-ZF beamforming, we  design the beamforming vectors in such a way that they  enforce the interference at the UEs to  zero. Actually, considering $\mathbf{g}_{_{k,\ell}}=\mathbf{H}_k \bPhi_{k}^{H} \mathbf{h}_{k,\ell}$ as the channel vector between the BS and the $\ell^{\rm th}$  blocked UE connected to $k^{\rm th}$ RIS and defining 
$\mathbf{Q}_1=\begin{bmatrix}
\mathbf{g}_{_{1,1}} \dots \mathbf{g}_{_{1,L_1}}   \dots \mathbf{g}_{_{K,1}} \dots \mathbf{g}_{_{K,L_{K}}}  \ \mathbf{h}_{d,1} \dots \mathbf{h}_{d,U_d}
\end{bmatrix}^{H},$
 the ZF beamforming matrix is given as 
 \begin{equation}
 \mathbf{W}_{\mathrm{ZF}}=\mathbf{Q}_1^{H}(\mathbf{Q}_1\mathbf{Q}_1^{H})^{-1},
 \label{eq:UEZF}
\end{equation} 
 which under the condition $M\geq (U_b+U_d)$, eliminates the inter-user interference \cite{MIMO2004}. Therefore, the SINR of the $\ell^{\rm th}$ blocked UE connected to the $k^{\rm th}$ RIS  becomes
\begin{align}
\label{eq_SINR_zf}
\mathrm{SINR}_{b,k,\ell}&=\dfrac{|\mathbf{g}_{k,\ell}^H \mathbf{Q}_1^{H}( \mathbf{Q}_1\mathbf{Q}_1^{H})^{-1} \mathbf{e}_k|^2}{\sigma^2_{k,\ell}}   
=\dfrac{|\mathbf{v}_k^H \mathbf{q}_{k,\ell}|^2}{\sigma^2_{k,\ell}},
\end{align}
where $\mathbf{q}_{k,\ell}={\rm diag}(\mathbf{h}_{k,\ell}^H) \mathbf{H}_k^H \mathbf{Q}_1^{H}( \mathbf{Q}_1\mathbf{Q}_1^{H})^{-1} \mathbf{e}_k  $,  $\mathbf{e}_{k}$ is the unit-base vector with non-zero value at the $k^{\rm th}$ component, 
$\mathbf{v}_k={\rm Vec}(\bPhi_k^{H})$ and operator ${\rm Vec}(.)$ returns the vector of the diagonal components of the diagonal matrix.

In order to  jointly maximize  the SINR of the blocked UEs, as a suboptimal solution, for $k\in \mathcal{K}$
we maximize the sum of the SINRs of the UEs connected to the $k^{\rm th}$ RIS, i.e.,
\begin{align}
&\max_{\mathbf{v}_k}  \sum_{\ell=1}^{L_k} |\mathbf{v}_k^H \mathbf{q}_{k,\ell}|^2 \equiv \max_{\mathbf{v}_k} \mathbf{v}_k^H \left(\sum_{\ell=1}^{L_k} \mathbf{q}_{k,\ell}\mathbf{q}_{k,\ell}^{H}\right)\mathbf{v}_k \label{prob_v_k} \\ &\,\,\mathrm{s.t.}\, |\mathbf{v}_k|^2 \leq 1 \quad \quad \quad \quad \mathrm{s.t.}\, |\mathbf{v}_k|^2 \leq 1, \notag
\end{align}
The optimum value of $ \mathbf{v}_k$ in \eqref{prob_v_k} is equal to the eigenvector of matrix $\sum_{\ell=1}^{L_k} \mathbf{q}_{k,\ell}\mathbf{q}_{k,\ell}^{H}$ corresponding to its maximum eigenvalue.
Thus, for $L_k=1$, the optimum value of $\phi_{k,i}$ is
\begin{align}\label{eq_phi_opt}
\hspace{-.1cm} \phi_{k,i}\hspace{-1mm}=\hspace{-1mm}-\sphericalangle(h^{*}_{k,1,i}[\mathbf{H}_{k}^{H} \mathbf{Q}_1^{H}( \mathbf{Q}_1\mathbf{Q}_1^{H})^{-1} \mathbf{e}_k]_{_{i}}), k\in \mathcal{K}, i\in \mathcal{N}.
\end{align}

In the following proposition, we obtain the asymptotic form of \eqref{eq_phi_opt} as $M \longrightarrow \infty$.  
\begin{proposition}\label{pre_asymp}
Considering the BS-UE-ZF beamforming in \eqref{eq:UEZF} and the asymptotic regime where $M\longrightarrow \infty$, when one UE is connected to the $k^{\rm th}$ RIS  the phase shifts of the RIS elements that maximize the SINR of the UEs can be 
found by solving  the following equation system:
\begin{equation}\label{eq_prep}
\phi_{k,i}=-\sphericalangle (h_{k,1,i}^{*} \sum_{\ell=1}^{N} R_{i,\ell} e^{-j \phi_{k,\ell}}  h_{k,1,\ell})
\   i\in \mathcal{N},
\end{equation}
where $R_{k,i,\ell}$ is the element $(i,\ell)$ of the correlation matrix  $\mathbf{R}_{k}=\mathbb{E}\{\mathbf{H}_{k}^{H}\mathbf{H}_{k}\}$.
\end{proposition}
\begin{proof}
Considering $\mathbf{Q}_1=\begin{bmatrix} 
\mathbf{g}_{_{1,1}}  \dots \mathbf{g}_{_{K,1}}  \ \mathbf{h}_{d,1} \dots \mathbf{h}_{d,U_d}
\end{bmatrix}^{H},$ for $k\in\mathcal{K}$ we have 
 \begin{align}
& \mathbf{H}_{k}^{H} \mathbf{Q}_1^{H}( \mathbf{Q}_1\mathbf{Q}_1^{H})^{-1} \mathbf{e}_k= \notag \\ &  \mathbf{H}_{k}^{H} \begin{bmatrix}
\mathbf{H}_1 \bPhi^{H}_1 \mathbf{h}_{1,1}  \dots  \mathbf{H}_{K} \bPhi^{H}_{K} \mathbf{h}_{K,1} \ \ \mathbf{h}_{d,1} \dots \mathbf{h}_{d,U_d}  
\end{bmatrix}
\\ &
\times \scriptsize{\begin{bmatrix}
\mathbf{g}_{_{1,1}}^H \mathbf{g}_{_{1,1}} & \dots & \mathbf{g}_{_{1,1}}^H \mathbf{g}_{_{K,1}} & \mathbf{g}_{_{1,1}}^H \mathbf{h}_{d,1} & \dots & \mathbf{g}_{_{1,1}}^H \mathbf{h}_{d,U_d}  \notag \\
\vdots & & \vdots & \vdots & & \vdots
\\
\mathbf{g}_{_{K,1}}^H \mathbf{g}_{_{1,1}} & \dots & \mathbf{g}_{_{K,1}}^H \mathbf{g}_{_{K,1}} & \mathbf{g}_{_{K,1}}^H \mathbf{h}_{d,1} & \dots & \mathbf{g}_{_{K,1}}^H \mathbf{h}_{d,U_d} 
\\
\\
\mathbf{h}_{d,1}^{H}\mathbf{g}_{_{1,1}}& \dots & \mathbf{h}_{d,1}^{H}\mathbf{g}_{_{K,1}} & \mathbf{h}_{d,1}^{H} \mathbf{h}_{d,1} & \dots & \mathbf{h}_{d,1}^{H} \mathbf{h}_{d,U_d}
\\
\vdots & & \vdots & \vdots & & \vdots
\\
\mathbf{h}_{d,U_d}^{H}\mathbf{g}_{_{1,1}}& \dots & \mathbf{h}_{d,U_d}^{H}\mathbf{g}_{_{K,1}} & \mathbf{h}_{d,U_d}^{H} \mathbf{h}_{d,1} & \dots & \mathbf{h}_{d,U_d}^{H} \mathbf{h}_{d,U_d}
\end{bmatrix}^{-1}} \hspace{-0.3cm}\mathbf{e}_{k}. \label{eq_pre_proof1}
 \end{align}
Also, if 
$\mathbf{H}_{k}=\begin{bmatrix}
\mathbf{h}^{\prime}_{k,1} & \dots & \mathbf{h}^{\prime}_{k,N}
\end{bmatrix}$,
as $M \longrightarrow\infty $,
we get $\mathbf{h}^{\prime H}_{k,\ell}\mathbf{h}^{\prime}_{n,m} \longrightarrow  \mathbb{E}\{\mathbf{h}^{\prime H}_{k,\ell}\mathbf{h}^{\prime}_{n,m} \}$ \cite{randommatrix}, and thus 
\begin{equation}
\label{eq_R_asymp}
\mathbf{H}_{k}^{H}\mathbf{H}_{n}\longrightarrow 
\begin{cases}
 \mathbf{R}_{k} & n=k
\\
\mathbf{0}_{N\times N} & n \neq k
\end{cases}.
\end{equation}
where $\mathbf{R}_{k}=\mathbb{E}\{\mathbf{H}_{k}^{H}\mathbf{H}_{k}\}$ and $\mathbf{0}_{N\times N}$ is an $N\times N$ zero matrix. Therefore, as $M\longrightarrow \infty$ we have
\begin{align}
 &\mathbf{H}_{k}^{H} \mathbf{Q}_1^{H}( \mathbf{Q}_1\mathbf{Q}_1^{H})^{-1} \mathbf{e}_k=
\begin{bmatrix}
\mathbf{0} \dots  \mathbf{R}_{k}\bPhi_{k}^{H}\mathbf{h}_{k,1}   \dots \mathbf{0}
\end{bmatrix} \notag  \\ 
 & \times \mathrm{diag}\Big( { \small \frac{1}{ \mathbf{h}_{1,1}^H \bPhi_1 \mathbf{R}_1 \bPhi_1^H \mathbf{h}_{1,1} },\dots, \frac{1}{ \mathbf{h}_{K,1}^H \bPhi_{K }\mathbf{R}_{K} \bPhi_{K}^H \mathbf{h}_{K,1} }}, \notag \\ & {\small\frac{1}{\mathbf{h}_{d,1}^{H} \mathbf{h}_{d,1}}, \dots, \frac{1}{\mathbf{h}_{d,U_d}^{H} \mathbf{h}_{d,U_d}}} \Big)\mathbf{e}_{k} 
  =\frac{\mathbf{R}_{k}\bPhi_{k}^{H}\mathbf{h}_{k,1} }{\mathbf{h}_{k,1}^H \bPhi_k \mathbf{R}_k \bPhi_k^H \mathbf{h}_{k} }
\label{eq_prep_proof2}
\end{align}
and thus  \eqref{eq_phi_opt} would be equal to
\begin{align}
	 \label{eq:optohi}
 \phi_{k,i}& =-\sphericalangle (\frac{1}{\mathbf{h}_{k,1}^H \bPhi_k \mathbf{R}_k \bPhi_k^H \mathbf{h}_{k,1}  }h_{k,1,i}^{*} \sum_{\ell=1}^{N} R_{i,\ell} e^{-j \phi_{k,\ell}}  h_{k,1,\ell})\nonumber \\ &=-\sphericalangle(h_{k,1,i}^{*} \sum_{\ell=1}^{N} R_{i,\ell} e^{-j \phi_{k,\ell}}  h_{k,1,\ell}), \ \ k\in \mathcal{K}, i\in \mathcal{N},
\end{align}
where \eqref{eq:optohi} is resulted from the fact that  $\mathbf{h}_{k,1}^H \bPhi_k \mathbf{R}_k \bPhi_k^H \mathbf{h}_{k,1}$ is a real and positive number.
\end{proof}
\begin{corollary}\label{corollary1}
	For the special case that  the channel vectors between the BS and each of the elements of the $k^{\rm th}$ RIS are independent, i.e., $\ \mathbf{R}_k=\mathbf{I}$, employing the BS-UE-ZF approach, 
	the  SINR of the  blocked UE connected to the $k^{\rm th}$ RIS is maximized by considering random phase shifts for the RIS elements, as $M\longrightarrow \infty$.
\end{corollary}
\begin{proof}
	By substituting $ \mathbf{R}_k$ with $\mathbf{I}$ in \eqref{eq_prep}, we get 
	$\phi_{k,i}=\phi_{k,i} \ \forall k,i$, meaning that $\phi_{k,i}$ can be randomly chosen.
	\end{proof}
\vspace{-0.4cm}
\subsection{BS-RIS-ZF Beamforming}
In this section, we consider the case that one UE is related to each RIS. We design a novel approach, i.e., BS-RIS-ZF, in which the beamforming vectors are derived to enforce the interference at  the RISs and direct UEs to zero\footnote{Note that the BS-RIS-ZF  beamforming approach can also be implemented for the general scenario; But, in this case,  it needs to design the phase shifts of the RISs to  cancel the interference which is out of scope of this letter.}.
Thus, the beamforming vectors should satisfy the following conditions for $k\in \mathcal{K}$ and $u\in \{1,\dots,U_d\}$:
 \begin{align}
 & \mathbf{H}_k^{H} \mathbf{w}_{b,k,1}=\mathbf{1}, \ \mathbf{H}_{m}^H \mathbf{w}_{b,k,1}=\mathbf{0}\ \forall  m \neq k, \ \mathbf{h}_{d,u}^H \mathbf{w}_{b,k,1}=0, \notag \\ &{ 
 \mathbf{h}_{d,u}^H \mathbf{w}_{d,u}=1, \ \mathbf{H}_{k}^H \mathbf{w}_{d,u}=\mathbf{0},\ \mathbf{h}_{d,n}^{H}\mathbf{w}_{d,u}=0} \ \forall  n \neq u.
 \end{align}
 
 By defining 
 $\mathbf{Q}_2 \triangleq \begin{bmatrix}
\mathbf{H}_1
\dots
\mathbf{H}_{K} \  \ \mathbf{h}_{d,1} \dots \mathbf{h}_{d,U_d}
 \end{bmatrix}^{H}$,
 the BS-RIS-ZF beamforming matrix can be written as
 \begin{equation}
 \mathbf{W}_{\rm ZF}^{\prime}=\mathbf{Q}_2^{H}(\mathbf{Q}_2\mathbf{Q}_2^{H})^{-1}\bGamma,
 \end{equation}
where 
$\bGamma=$

 $\scriptsize{\begin{bmatrix}
1 & \dots & 1  & 0 & & & \dots & & 0& \\
0 & \dots & 0 & 1& \dots & 1 &0 &\dots & 0 & \\
\vdots & \ddots  &\ddots &\ddots &\ddots &\ddots &\ddots &\ddots & \vdots  & \mathbf{0}_{K \times U_d}   \\
0 & & &  \hdots   & &0 &   1 & \hdots&1 &\\
 &  & & & {\mathbf{0} _{U_d\times NK }}& & && & \mathbf{I}_{U_d}
\end{bmatrix}}^T$.
Thus, the SINR of the blocked UE connected to  $k^{\rm th}$ RIS  is 
\begin{equation}
\label{eq_SINR3}
\hspace{-1mm}\mathrm{SINR}_{b,k,1}\hspace{-1.05mm}=\hspace{-1.05mm}\dfrac{| \sum\limits_{i=1}^{N} e^{j \phi_{i,k}} h_{k,1,i}^{*}  [\mathbf{H}_{k}^{H} \mathbf{Q}_2^{H} (\mathbf{Q}_2\mathbf{Q}_2^{H})^{-1} \bGamma \mathbf{e}_{k}]_{_{i}}|^2 }{\sigma^{2}_{k}}.\hspace{-1.5mm}
\end{equation}

Similar to the previous subsection, the optimal phase shifts of the  RIS elements  make the phases of the summation terms in   the numerator  of the SINR ratio equal to zero; thus,
\begin{equation}\label{eq_phi_opt2}
\hspace{-1mm} \phi_{k,i}\hspace{-1mm}=\hspace{-1mm}-\sphericalangle( h_{k,1,i}^{*} [\mathbf{H}_{k}^{H} \mathbf{Q}_2^{H} (\mathbf{Q}_2\mathbf{Q}_2^{H})^{-1}\bGamma \mathbf{e}_{k}]_{_{i}}) \hspace{1mm}  k\in\mathcal{K}, i\in\mathcal{N}\hspace{-2.5mm}
\end{equation}

The major benefit of the BS-RIS-ZF beamforming is the availability of the closed-form equations for the RIS phase shifts. However, we must note that in the BS-RIS-ZF  approach we need to satisfy $M>NK+U_d$ to completely remove the interference, while increasing the number of the BS antennas increases the complexity of the channel estimation \cite{channel_estimation}. Therefore, there exists a trade-off between the complexity of the phase shift design and channel estimation.

Now, in the following, we evaluate the asymptotic expressions for the optimal phase shifts and   maximum $\mathrm{SINR}_{b,k,1}$ as $M \longrightarrow \infty$.
 \begin{proposition}
 Considering the proposed BS-RIS-ZF beamforming, as $M \longrightarrow \infty$, the RIS  phase shifts that maximize the SINR of the blocked UE connected to the $k^{\rm th}$ RIS are
 \begin{equation}
 \label{eq_phase_aymp}
  \phi_{k,i}=-\sphericalangle(h_{k,i}^{*} f_{k,i}), \ i\in\mathcal{N}
 \end{equation}
  and the maximum value of  $\mathrm{SINR}_{b,k,1}$  is 
 \begin{equation}
  \label{eq_SINR_asymp}
 \mathrm{SINR}_{b,k,1}^{*}=\frac{1}{\sigma^{2}_k}(\sum_{i=1}^{N} |f_{k,i}| |h_{k,i}|)^2,
 \end{equation}
 where $f_{k,i}$ is the $i^{\rm th}$ element of the vector
  $$\mathbf{f}_{k}\hspace{-0.25em}=\hspace{-0.25em}\begin{bmatrix}
 \mathbf{0}_{N \times (k-1)N} \ \mathbf{R}_{k} \ \mathbf{0}_{N \times ((K-k)N+U_d)}
 \end{bmatrix}\hspace{-0.25em}  $$ $$
\times {\mathrm{blkdiag}}(
 \mathbf{R}_{1}^{-1}, 
\dots,  
  \mathbf{R}_{K}^{-1}, 
 \mathbf{I}_{U_d}
)\bGamma\, \mathbf{e}_{k},
$$ and $\mathrm{blkdiag}(.)$ returns a block diagonal matrix constructed by the arguments.
 \end{proposition}

 \begin{proof}
As $M\longrightarrow \infty$, employing \eqref{eq_R_asymp}, we have
 \begin{equation}\label{eq_HH^H}
 (\mathbf{Q}_2\mathbf{Q}_2^{H})^{-1}=  {\mathrm{blkdiag}}(
 \mathbf{R}_{1}^{-1}, 
\dots,  
  \mathbf{R}_{K}^{-1},
 \mathbf{I}_{U_d}
),
 \end{equation}
 and
\begin{equation}\label{eq_HkH}
 \mathbf{H}_{k}^{H}\mathbf{Q}_2^{H}=\begin{bmatrix}
 \mathbf{0}_{N \times (k-1)N} & \mathbf{R}_{k} & \mathbf{0}_{N \times ((K-k)N+U_d)}
 \end{bmatrix}.
\end{equation}
 Thus, substituting \eqref{eq_HH^H} and \eqref{eq_HkH} in \eqref{eq_phi_opt2}, the asymptotic expressions for the optimal phase shifts  are as in \eqref{eq_phase_aymp}. Then,
substituting the phase shifts obtained by \eqref{eq_phase_aymp}, \eqref{eq_HH^H}, and \eqref{eq_HkH} in \eqref{eq_SINR3}, the asymptotic value of maximum $\mathrm{SINR}_k$ would be obtained as in \eqref{eq_SINR_asymp}.
 \end{proof}

\textbf{Analysis of the rank of matrix $\mathbf{Q}_2$:} 
Regarding that the performance of the BS-RIS-ZF approach  can be affected by the rank of matrix $\mathbf{Q}_2$, we analyze the rank of this matrix in this section. Assuming that $\mathbf{H}_k= \mathbf{F}_k \mathbf{D}_{k}$, where $ \mathbf{F}_k$ is  an $M \times N$ matrix with i.i.d. normal random variable elements and $\mathbf{D}_k=\mathbf{R}_k^{1/2}$, the matrix $\mathbf{Q}_2$ can be rewritten as
 $\begin{bmatrix}
\mathbf{F}_1  \mathbf{D}_{1} 
 \dots 
\mathbf{F}_K  \mathbf{D}_{K} \  \ \mathbf{H}_{d}
 \end{bmatrix}^{H},$ where $\mathbf{H}_{d}=[\mathbf{h}_{d,1} \dots \mathbf{h}_{d,U_d} ]$.  Considering the fact that the elements of the random matrices $\mathbf{F}_k$, $\ k \in \mathcal{K}$, 
 and $\mathbf{H}_{d}$ are  independent of each other, we conclude that the rank of matrix $\mathbf{Q}_2$ is equal to $\sum_{k=1}^{K} \mathrm{rank}( \mathbf{D}_{k}^{H} \mathbf{F}_k^{H})+\mathrm{rank}(\mathbf{H}_{d})$. 
Also, from \cite{strang}, we have
  $$\mathrm{rank}( \mathbf{D}_{k}^{H} \mathbf{F}_k^{H})\leqslant\min \{ \mathrm{rank}( \mathbf{D}_{k}^{H}),  \mathrm{rank}( \mathbf{F}_k^{H})\},$$ and  $\mathrm{rank}( \mathbf{F}_k^{H}) \leqslant \min \{M,N\}$.
 Hence, assuming that $N\leqslant M$, we get  $\mathrm{rank}( \mathbf{D}_{k}^{H} \mathbf{F}_k^{H})\leqslant  \mathrm{rank}( \mathbf{D}_{k}^{H})$ and thus,
 $\mathrm{rank}( \mathbf{Q}_2)\leqslant \sum_{k=1}^{K} \mathrm{rank}(\mathbf{D}_{k}^{H})+\mathrm{rank}(\mathbf{H}_{d}).$ 
 Moreover, the elements of $\mathbf{H}_{d}$   are independently distributed and thus, $\mathrm{rank}(\mathbf{H}_{d})=\min \{M, U_d\}$. Therefore, assuming that $U_d<M$ we have $\mathrm{rank}( \mathbf{Q}_2)\leqslant \sum_{k=1}^{K} \mathrm{rank}(\mathbf{D}_{k}^{H})+U_d$. Consequently, we can conclude that the correlation of the BS to RIS channels can restrict the rank of matrix $\mathbf{Q}_2$.
 \vspace{-0.4cm}
 \subsection{  Computational Complexity.} 
 The number of multiplications required to compute the beamforming vectors and the RIS phase shifts in the BS-UE-ZF approach  is equal to $U_b((U_b+U_d)^3+2M(U_b+U_d)^2+MN(U_b+U_d)+MN^2+MN+1)$, while   for the BS-RIS-ZF method,  it is equal to $U_b((NU_b+U_d)^3+(2M+U_b+K+d)(NU_b+U_d)^2+MN(NU_b+U_d)+1)$. We observe that the complexity of both of the approaches with respect to $M$ is of the order of $\mathcal{O}(M)$. Also,  their complexities with respect to $N$ are  of the order of $\mathcal{O}(N)$ and $\mathcal{O}(N^3)$, respectively. Therefore, the complexity of the BS-RIS-ZF approach is more sensitive to $N$.
 \section{Simulation Results}
 In this section, we conduct simulations to verify the performance  of the proposed beamforming and phase shift design approaches.  In this regard, we use the channel estimation approach proposed in \cite{channel_estimation}  and compare the results with the case of perfect channel estimation.
 The RIS elements correlation model is adopted from \cite{RRIS}.  The minimum distance between the elements of each of the RISs  is  equal to $d=\lambda$
 and $d=\lambda/4$ for the cases of the  i.i.d. and correlated channels, respectively, where $\lambda$ is the wavelength.  Also, the area of each RIS element is equal to $A=d^2$ and  $\mu  \lambda^2=-75$dB where $\mu$ is the average intensity attenuation \cite{RRIS}.
 Moreover, the carrier frequency is  $f=1800$~MHz and   the power spectral density of the AWGN at the   UEs is equal to $-174 ~\rm{dBm}/\rm{Hz}$. 
     \begin{figure}[t] \centering
    \begin{subfigure}{0.32\textwidth}\includegraphics[width=1\columnwidth]{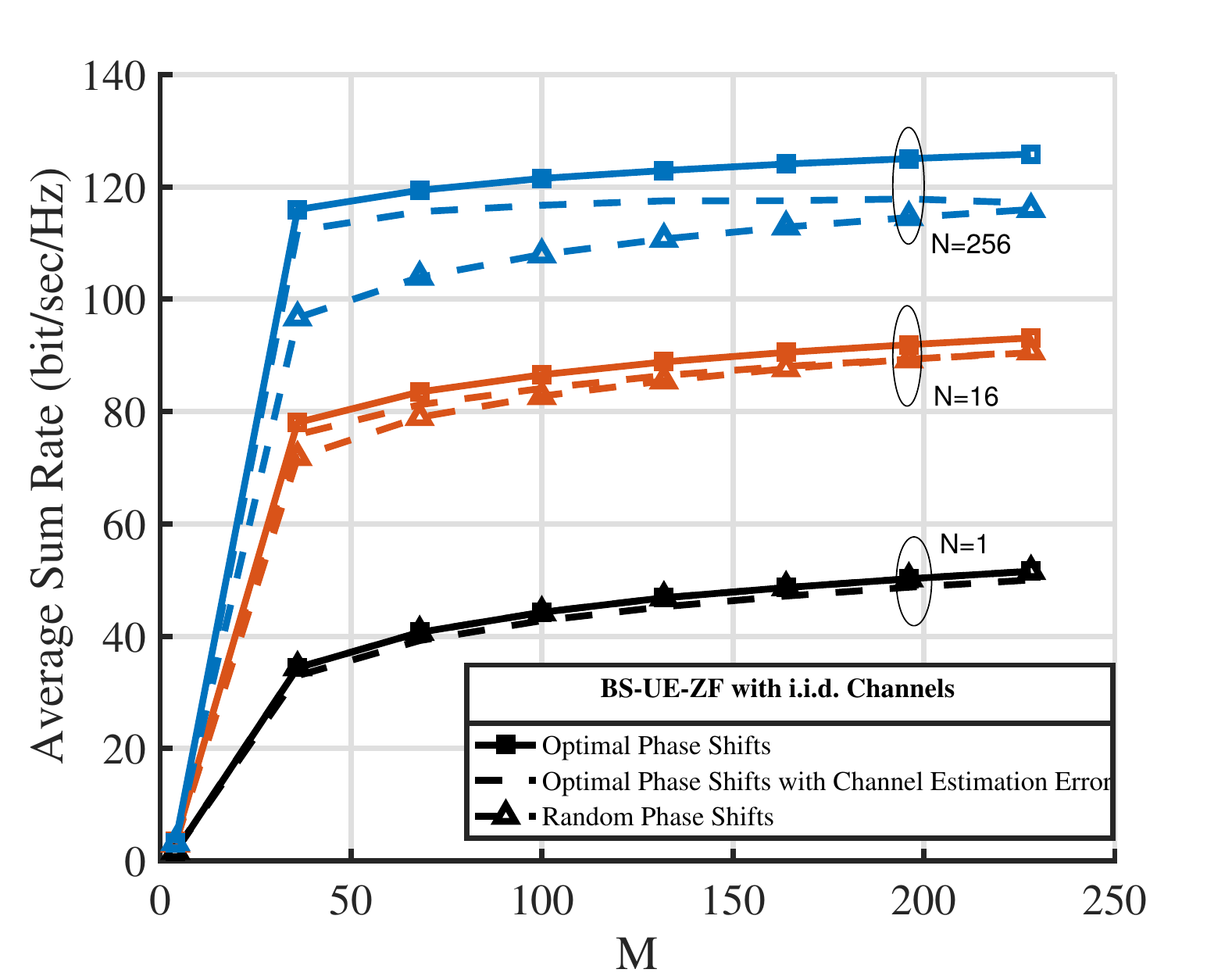}\caption{}\label{fig_BStoUser_iid}
    \end{subfigure}
 \begin{subfigure}{0.32\textwidth}\includegraphics[width=1\columnwidth]{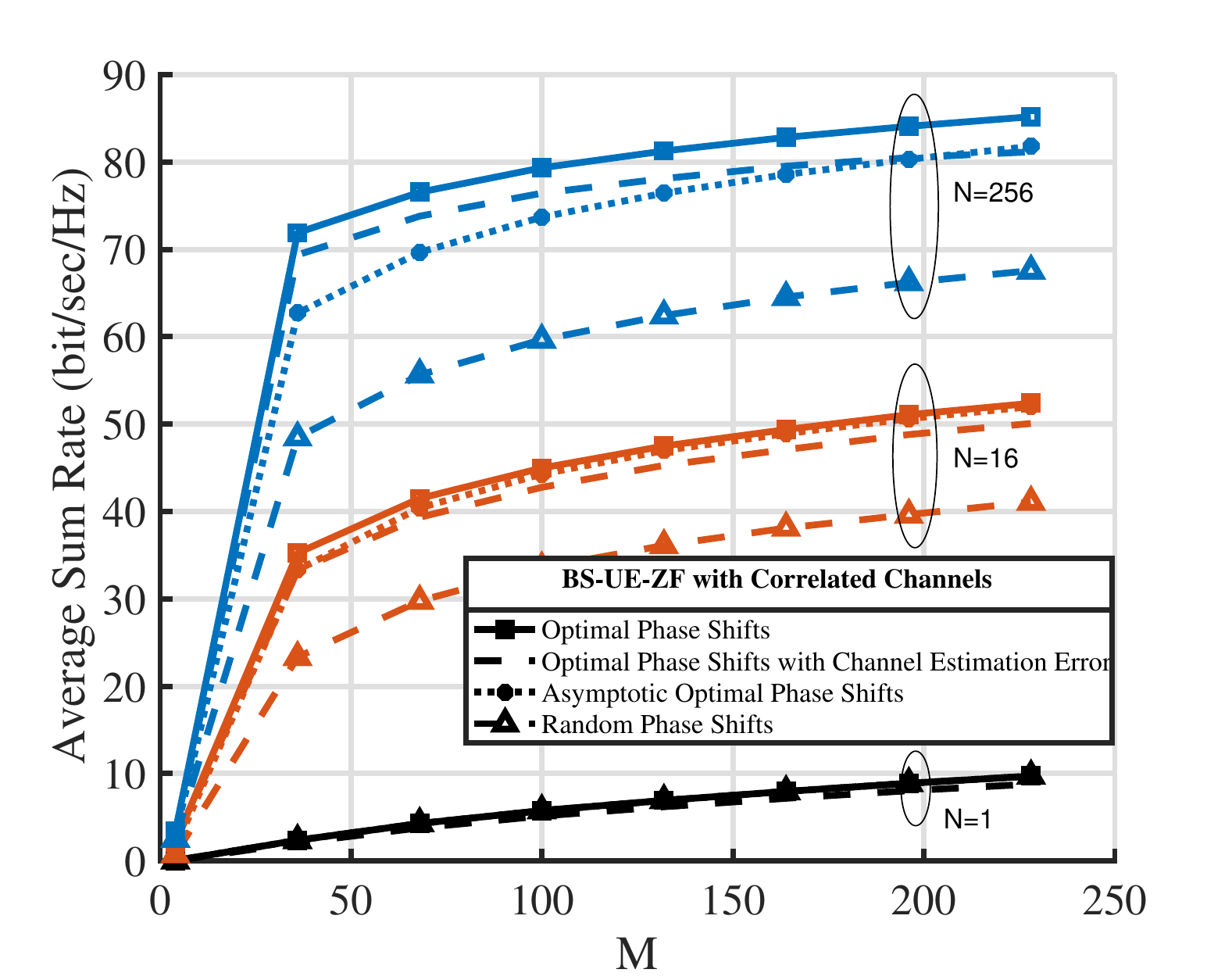}\caption{} \label{fig_BStoUser_cor}
     \end{subfigure}
 \caption{ Average sum rate exploiting the BS-UE-ZF  approach vs. the number of the BS antennas ($M$), $K=4$, $U_d=2$ (a): i.i.d.  (b): Correlated  channels.} \label{fig_BStoUser}
    \end{figure}
 
 In Fig. \ref{fig_BStoUser}, we illustrate the average sum rate versus the number of the BS antennas, for various phase shift design approaches and the number of RIS elements in the BS-UE-ZF beamforming scenario,  considering both the   i.i.d. and correlated channels.  It can be noticed  that in both cases increasing the number of  the elements in the RISs  improves the average sum rate, and  $N=1$ leads to the lowest average sum rate. 
 In Fig.\ref{fig_BStoUser_iid}, we observe that for the i.i.d.  channels, by increasing the number of the BS antennas the average sum rate of random phase shifts  tends to the average sum rate of the optimal ones (see Corollary \ref{corollary1}).  In Fig. \ref{fig_BStoUser_cor}, we observe that in the case of the correlated  channels, by increasing the number of the BS antennas, the  average sum rate of  the asymptotic phase shifts converges to the average sum rate of the optimal phase shifts. Moreover,  both of these phase shift design approaches achieve a much more average sum  rate compared to  random phase shifts. 
 
 In Fig. \ref{fig_BStoIRS}, we plot the average sum rate versus the number of the BS antennas for 
 the various phase shift design approaches in the BS-RIS-ZF beamforming scenario, and also depict the asymptotic average sum rate  obtained by the asymptotic SINRs. We observe that for both the i.i.d. and correlated channels,  the  rates of the asymptotic phase shifts and the asymptotic SINRs  accurately track  the optimal phase shifts rate.
Moreover, it can be observed that in the case of  the i.i.d.  channels when the condition $M\geqslant KN+U_d$ is satisfied, the average sum rates start to increase,  and the correlated channels do not have   a  restrictive impact on the start point of rate arising. 

Furthermore, in  Fig. \ref{fig_BStoUser} and Fig. \ref{fig_BStoIRS}, we illustrate 
 the curves corresponding to the case that  we have channel estimation error  (the channel estimation approach in \cite{channel_estimation} is used).   We observe that in both of the i.i.d. and correlated channels there is not a significant  reduction in  the sum rate of the BS-UE-ZF and BS-RIS-ZF approaches when we have some error in the channel estimation. 
      \begin{figure}[t] \centering
    \begin{subfigure}{0.32\textwidth}\includegraphics[width=1\columnwidth]{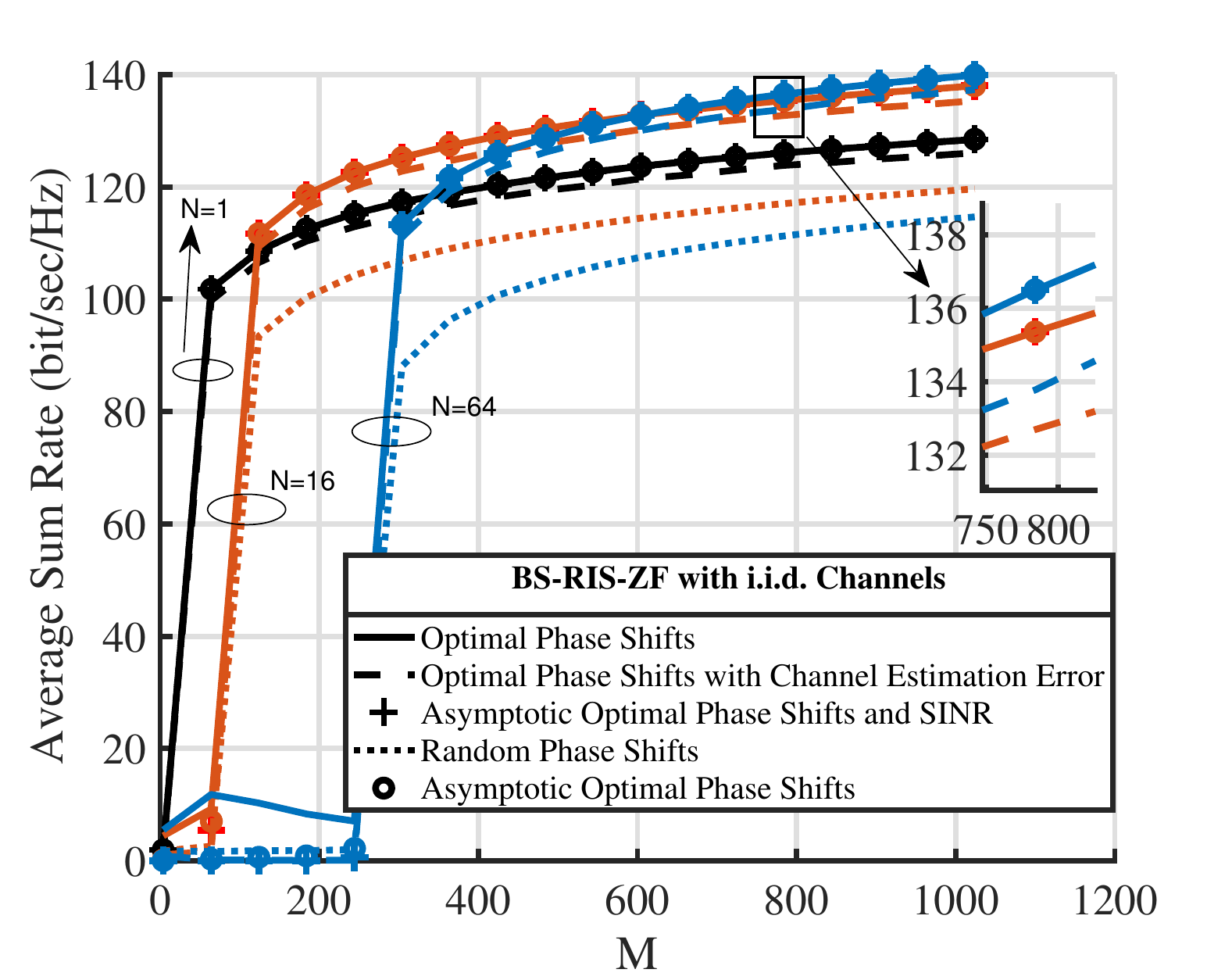}\caption{}\label{fig_BStoIRS_iid}
    \end{subfigure}
 \begin{subfigure}{0.32\textwidth}\includegraphics[width=1\columnwidth]{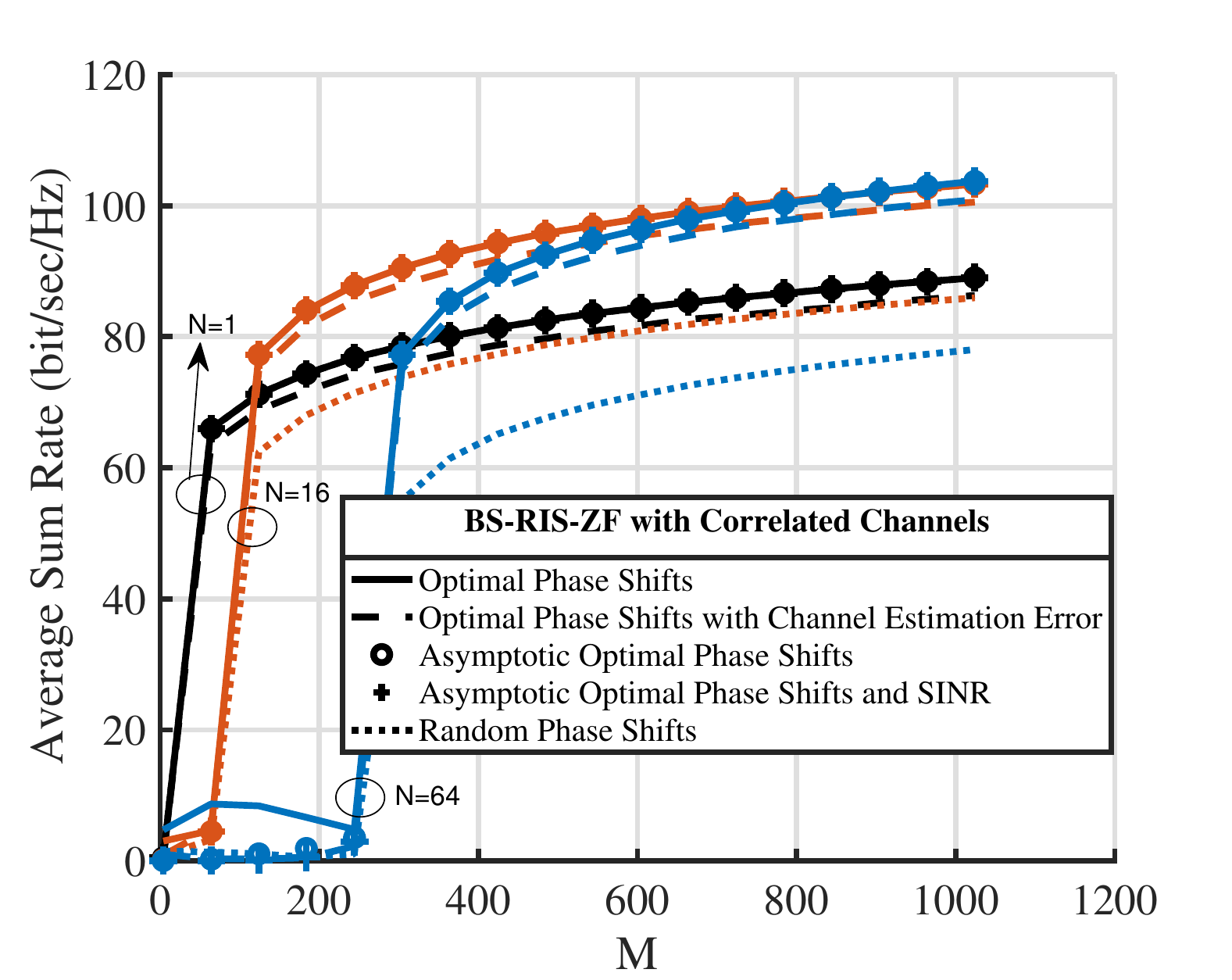}\caption{}
 \label{fig_BStoIRS_cor}
     \end{subfigure}
 \caption{ Average sum rate exploiting the BS-RIS-ZF  approach vs. the number of the BS antennas ($M$), $K=4$, $U_d=2$ (a): i.i.d.  (b): Correlated  channels. }
   \label{fig_BStoIRS}
    \end{figure}
    \vspace{-0.5cm}
 \section{Conclusion}
In this letter, we considered  a multi-RIS, multi-user massive  MIMO system and investigated the SINR maximization problem for the asymptotic scenario where the number of the BS antennas tends to infinity. We  examined two ZF beamforming approaches, i.e.,  BS-UE-ZF and BS-RIS-ZF which null the interference at the UEs and  the RISs, respectively. For each of the proposed methods, we  obtained the optimal phase shifts of the  elements of the RISs that maximize the SINR of the UEs.  Considering the  BS-UE-ZF beamforming approach, we showed that when the channels of the RIS elements are independent, random  phase shifts can achieve the maximum SINR for an asymptotic large number of  BS antennas. For the BS-UE-ZF beamforming method, the simulation results showed that  the asymptotic expressions of the RIS phase shifts  achieve the rate of the optimal phase shifts, even for a small number of the BS antennas. 
\ifCLASSOPTIONcaptionsoff
  \newpage
\fi
\vspace{-0.2cm}
\bibliographystyle{IEEEtran}
\bibliography{ref}
\end{document}